\documentclass[10pt,twocolumn,twoside]{IEEEtran}

\IEEEoverridecommandlockouts                              	

\usepackage{amssymb,amsmath,amsfonts,amsthm}
\usepackage{algorithm,algorithmic}
\usepackage{cite}
\usepackage{float}
\usepackage{epsfig}
\usepackage{graphicx}
\usepackage{graphics}
\usepackage{subfigure}
\usepackage{url}
\usepackage{xspace}
\usepackage[usenames,dvipsnames]{color}
\usepackage{soul}
\usepackage{mathtools}

\floatstyle{ruled}
\newfloat{model}{H}{mod}
\floatname{model}{\footnotesize Model}
\newfloat{notatio}{H}{not}
\floatname{notatio}{\footnotesize Notation}

\newenvironment{varalgorithm}[1]
  {\algorithm\renewcommand{\thealgorithm}{#1}}
  {\endalgorithm}

\newenvironment{list3}{
	\begin{list}{$\bullet$}{%
			\setlength{\itemsep}{0.05cm}
			\setlength{\labelsep}{0.2cm}
			\setlength{\labelwidth}{0.3cm}
			\setlength{\parsep}{0in} 
			\setlength{\parskip}{0in}
			\setlength{\topsep}{0in} 
			\setlength{\partopsep}{0in}
			\setlength{\leftmargin}{0.22in}}}
	{\end{list}}

\newenvironment{list4}{
	\begin{list}{$\bullet$}{%
			\setlength{\itemsep}{0.05cm}
			\setlength{\labelsep}{0.2cm}
			\setlength{\labelwidth}{0.3cm}
			\setlength{\parsep}{0in} 
			\setlength{\parskip}{0in}
			\setlength{\topsep}{0in} 
			\setlength{\partopsep}{0in}
			\setlength{\leftmargin}{0.16in}}}
	{\end{list}}

\newenvironment{list4a}{
	\begin{list}{$\bullet$}{%
			\setlength{\itemsep}{0.05cm}
			\setlength{\labelsep}{0.2cm}
			\setlength{\labelwidth}{0.3cm}
			\setlength{\parsep}{0in} 
			\setlength{\parskip}{0in}
			\setlength{\topsep}{0in} 
			\setlength{\partopsep}{0in}
			\setlength{\leftmargin}{0.16in}}}
	{\end{list}}

\newenvironment{list5}{
	\begin{list}{$\bullet$}{%
			\setlength{\itemsep}{0.05cm}
			\setlength{\labelsep}{0.2cm}
			\setlength{\labelwidth}{0.3cm}
			\setlength{\parsep}{0in} 
			\setlength{\parskip}{0in}
			\setlength{\topsep}{0in} 
			\setlength{\partopsep}{0in}
			\setlength{\leftmargin}{0.18in}}}
	{\end{list}}

\usepackage{url,changebar,bm,xspace,dsfont}
\let\mathbb=\mathds %

\newtheorem{theorem}{Theorem}

\newtheorem{defn}{Definition}

\newtheorem{assum}{Assumption}

\newtheorem{remark}{Remark}

\ifCLASSINFOpdf
 \else
\fi

\usepackage{amsmath,amsfonts,amssymb,amscd}
\usepackage{capt-of}
\usepackage{color}
\usepackage{cite}
\usepackage{verbatim}
\usepackage{hyperref}

\hypersetup{
    colorlinks = true,
    linkcolor = [rgb]{0,0,1},
    anchorcolor = [rgb]{0,0,1},
    citecolor = [rgb]{0.9,0.5,0},
    filecolor = [rgb]{0,0,1},
    urlcolor = [rgb]{0.9,0.5,0},
    bookmarksopen = true,
    bookmarksnumbered = true,
    breaklinks = true,
    linktocpage,
    colorlinks = true,
    linkcolor = [rgb]{0.2,0.6,0.2},
    urlcolor  = [rgb]{0,0,1},
    citecolor = [rgb]{0.9,0.5,0},
    anchorcolor = [rgb]{0.2,0.6,0.2},
}

\newtheorem{lemma}{\bfseries Lemma}

\begin{document}

\title{\LARGE \bf Optimal CPU Scheduling in Data Centers via \\ a Finite-Time Distributed Quantized Coordination Mechanism}

\author{Apostolos~I.~Rikos, Andreas Grammenos, Evangelia Kalyvianaki, \\ Christoforos N. Hadjicostis, Themistoklis Charalambous, and Karl~H.~Johansson
\thanks{Apostolos~I.~Rikos and K.~H.~Johansson are with the Division of Decision and Control Systems, KTH Royal Institute of Technology, SE-100 44 Stockholm, Sweden. E-mails: {\tt \{rikos,kallej\}@kth.se}.}
\thanks{A. Grammenos is with the Department of Computer Science and Technology, University of Cambridge, Cambridge, and the Alan Turing Institute, London, UK.  E-mail:{\tt~ag926@cl.cam.ac.uk}.}
\thanks{E. Kalyvianaki is with the Department of Computer Science and Technology, University of Cambridge, Cambridge, UK.  E-mail:{\tt~ek264@cl.cam.ac.uk}.}
\thanks{C. N. Hadjicostis is with the Department of Electrical and Computer Engineering, University of Cyprus, 1678 Nicosia, Cyprus.  E-mail:{\tt~chadjic@ucy.ac.cy}.}
\thanks{T. Charalambous is with the Department of Electrical Engineering and Automation, School of Electrical Engineering, Aalto University, 02150 Espoo, Finland.  E-mail:{\tt~themistoklis.charalambous@aalto.fi}.}
}

\maketitle
\thispagestyle{empty}
\pagestyle{empty}

\begin{abstract}
In this paper we analyze the problem of optimal task scheduling for data centers. 
Given the available resources and tasks, we propose a fast distributed iterative algorithm which operates over a large scale network of nodes and allows each of the interconnected nodes to reach agreement to an optimal solution in a finite number of time steps. 
More specifically, the algorithm (i) is guaranteed to converge to the exact optimal scheduling plan in a finite number of time steps and, (ii) once the goal of task scheduling is achieved, it exhibits distributed stopping capabilities (i.e., it allows the nodes to distributely determine whether they can terminate the operation of the algorithm). 
Furthermore, the proposed algorithm operates exclusively with quantized values (i.e., the information stored, processed and exchanged between neighboring agents is subject to deterministic uniform quantization) and relies on event-driven updates (e.g., to reduce energy consumption, communication bandwidth, network congestion, and/or processor usage). 
We also provide examples to illustrate the operation, performance, and potential advantages of the proposed algorithm. 
Finally, by using extensive empirical evaluations through simulations we show that the proposed algorithm exhibits state-of-the-art performance. 
\end{abstract}


\section{Introduction}\label{intro}

Modern Clouds infrastructure comprises a network of data centers, each containing thousands of server machines. Resource management in data centers is the procedure of allocating resources (e.g., CPU, memory, network bandwidth and disk space) to workloads such that their performance objectives are satisfied, given the available resources.

Resource allocation is inherently an optimization problem. %
However, solving it as such is challenging due to the scale and heterogeneity of the infrastructure and the dynamic nature of resource requirements of incoming and existing workloads. 
Centrally gathering all the required performance data from thousands of servers and running workloads, and solving the problem by a single solver is not ideal as gathered data becomes obsolete by the time the optimization is solved. For this reason, there has been recent interest towards practical distributed schedulers for solving this problem hierarchically. 
However, most of the proposed approaches employ heuristics that solve the problem approximately; see, e.g.,~\cite{mao_learning_2019, boutin_apollo_2014}. 

Recently, there has been a surge on distributed optimization, due to the wide variety of applications requiring related solutions ranging from distributed estimation to machine learning~\cite{Rabbat:2018IEEEProceedings,Tao:2019Survey}. 
Most of the works in the literature consider distributed solutions with asymptotic convergence which assume that the messages/quantities exchanged among nodes in the network are real numbers and therefore converge within some error \cite{2015_Alej_Hadj}. 
In several practical occasions, however, 
the quantities exchanged, such as scheduled tasks in CPU allocation, take discrete values. 
In addition, in many applications, such as in resource management in data centers, it is desirable to conclude the optimization in a finite number of steps via the exchange of quantized values, so that the exact solution is calculated and then applied.

In this paper, we focus on balancing the CPU utilization  across data center servers by carefully deciding how to allocate CPU resources to workloads in a distributed fashion. We further take into consideration that the allocated resources take discrete (quantized) values. 
We propose a distributed algorithm that solves and terminates the optimization problem in a finite number of steps using quantized values. 
Even though the proposed algorithm could be adopted in a wide variety of applications, here, we discuss it within the context of resource management in Cloud infrastructures. 
The main contributions of the paper are the following.
\begin{list5}
\item We present a distributed algorithm that solves the optimization problem in a finite number of time steps using quantized values. 
\item Then, we deploy a distributed stopping mechanism in order to terminate its operation, and hence the distributed optimization problem, in a finite number of time steps. It is the first distributed stopping mechanism for quantized average consensus algorithms. 
\item We provide an upper bound on the number of time steps needed for convergence based on properties of primitive matrices. 
The convergence time relies on connectivity (which is determined by the diameter of the network), rather than the size of the network. 
\item Simulations demonstrate that the proposed algorithm is suitable for large-scale networks, such as data centers. 
\end{list5}

The problem of providing a distributed solution to the resource coordination problem on a strongly connected digraph has been studied in the literature (see, e.g., \cite{2015_Alej_Hadj, 2013:CDC1Themis}), but for real values and not in an optimization context.
Our paper is a major departure from the current literature which mainly comprises distributed algorithms which operate with real values and exhibit asymptotic convergence within some error. 
Utilization of quantized values allows a more efficient usage of network resources, while finite time convergence allows calculation of the exact solution without any error. 
Our presented algorithm combines both characteristics and aims to pave the way for the use of fast bandwidth-efficient finite time algorithms which operate solely with quantized values over resource allocation problems.

\section{Notation and Preliminaries}\label{sec:preliminaries}

\subsection{Notation}

The sets of real, rational, integer and natural numbers are denoted by $ \mathbb{R}, \mathbb{Q}, \mathbb{Z}$ and $\mathbb{N}$, respectively. 
Symbols $\mathbb{Z}_{\geq 0}$ ($\mathbb{Z}_{>0}$) denote the sets of nonnegative (positive) integer numbers, while $\mathbb{Z}_{\leq 0}$ ($\mathbb{Z}_{<0}$) denote the sets of nonpositive (negative) integer numbers. 
For any real number $a \in \mathbb{R}$, the floor $\lfloor a \rfloor$ denotes the greatest integer less than or equal to $a$ while the ceiling $\lceil a \rceil$ denotes the least integer greater than or equal to $a$. 
Vectors are denoted by small letters, matrices are denoted by capital letters and the transpose of a matrix $A$ is denoted by $A^T$. 
For a matrix $A\in \mathbb{R}^{n\times n}$, the entry at row $i$ and column $j$ is denoted by $A_{ij}$.
By $\mathbf{1}$ we denote the all-ones vector and by $I$ we denote the identity matrix (of appropriate dimensions). 

Consider a network of $n$ ($n \geq 2$) nodes communicating only with their immediate neighbors. 
The communication topology is captured by a directed graph (digraph) defined as $\mathcal{G}_d = (\mathcal{V}, \mathcal{E})$. 
In digraph $\mathcal{G}_d$, $\mathcal{V} =  \{v_1, v_2, \dots, v_n\}$ is the set of nodes, whose cardinality is denoted as $n  = | \mathcal{V} | \geq 2 $, and $\mathcal{E} \subseteq \mathcal{V} \times \mathcal{V} - \{ (v_j, v_j) \ | \ v_j \in \mathcal{V} \}$ is the set of edges (self-edges excluded) whose cardinality is denoted as $m = | \mathcal{E} |$. 
A directed edge from node $v_i$ to node $v_j$ is denoted by $m_{ji} \triangleq (v_j, v_i) \in \mathcal{E}$, and captures the fact that node $v_j$ can receive information from node $v_i$ (but not the other way around). 
We assume that the given digraph $\mathcal{G}_d = (\mathcal{V}, \mathcal{E})$ is \textit{strongly connected}. 
This means that for each pair of nodes $v_j, v_i \in \mathcal{V}$, $v_j \neq v_i$, there exists a directed \textit{path}\footnote{A directed \textit{path} from $v_i$ to $v_j$ exists if we can find a sequence of nodes $v_i \equiv v_{l_0},v_{l_1}, \dots, v_{l_t} \equiv v_j$ such that $(v_{l_{\tau+1}},v_{l_{\tau}}) \in \mathcal{E}$ for $ \tau = 0, 1, \dots , t-1$.} from $v_i$ to $v_j$. 
Furthermore, the diameter $D$ of a digraph is the longest shortest path between any two nodes $v_j, v_i \in \mathcal{V}$ in the network. 
The subset of nodes that can directly transmit information to node $v_j$ is called the set of in-neighbors of $v_j$ and is represented by $\mathcal{N}_j^- = \{ v_i \in \mathcal{V} \; | \; (v_j,v_i)\in \mathcal{E}\}$. 
The cardinality of $\mathcal{N}_j^-$ is called the \textit{in-degree} of $v_j$ and is denoted by $\mathcal{D}_j^-$. 
The subset of nodes that can directly receive information from node $v_j$ is called the set of out-neighbors of $v_j$ and is represented by $\mathcal{N}_j^+ = \{ v_l \in \mathcal{V} \; | \; (v_l,v_j)\in \mathcal{E}\}$. 
The cardinality of $\mathcal{N}_j^+$ is called the \textit{out-degree} of $v_j$ and is denoted by $\mathcal{D}_j^+$. 

\subsection{Data Center and Workload Modelling}\label{systemModel}

We model a data center as a set $\mathcal{V}$ of server compute nodes, each denoted by $v_{i}\in\mathcal{V}$, which also operate as resource schedulers; this is a standard practice in modern data centers. 
All participating schedulers are usually interconnected with undirected communication links and, thus, the network topology forms a connected undirected graph. 
Nevertheless, our results are suitable for digraphs as well and, for this reason, hereafter we consider digraphs. 

A job is defined as a group of tasks, 
and $\mathcal{J}$ denotes the set of all jobs to be scheduled. 
Each job $b_{j} \in \mathcal{J}$, $j\in\{1,\ldots, |\mathcal{J}| \}$, requires $\rho_{j}$ cycles to be executed.
The estimated amount of resources (i.e., CPU cycles) needed for each job is assumed to be known before the optimization starts. A job task could require resources ranging from 1 to $\rho_j$ cycles, and the total sum of resources for all tasks of the same job is equal to $\rho_j$ cycles.
The total workload due to the jobs arriving at node $v_i$ is denoted by $l_i$. 
The time horizon $T_{h}$ is defined as the time period for which the optimization is considering the jobs to be running on the server nodes, before the next optimization decides the next allocation of resources. 
Hence, in this setting, the CPU capacity of each node, considered during the optimization, is computed as
$\pi_i^{\max} \coloneqq c_i T_h$, where $c_i$ is the sum of all clock rate frequencies of all processing cores of node $v_{i}$ given in cycles/second. The CPU availability for node $v_{i}$ at optimization step $m$ (i.e., at time $mT_{h}$) is given by $\pi_i^{\mathrm{avail}}[m] \coloneqq \pi_i^{\max} - u_{i}[m]$, where $u_{i}[k]$ is the number of unavailable/occupied cycles due to predicted or known utilization from already running tasks on the server over the time horizon $T_{h}$ at step $m$.

\begin{assum}
    We assume that the time horizon is chosen such that the total amount of resources demanded at a specific optimization step $m$, denoted by $\rho[m] \coloneqq \sum_{b_j[m] \in \mathcal{J}[m]} \rho_{j}[m]$, is smaller than the total capacity of the network available, given by $\pi^{\mathrm{avail}}[m] \coloneqq \sum_{v_{i}\in \mathcal{V}} \pi_i^{\mathrm{avail}}[m]$, i.e., $\rho[m] \leq \pi^{\mathrm{avail}}[m]$. 
\end{assum}

This assumption indicates that there is no more demand than the available resources. This assumption is realistic, since the time horizon $T_h$ can be chosen appropriately to fulfill the requirement. In case this assumption is violated, the solution will be that all resources are being used and some workloads will not be scheduled, due to lack of resources, but how to handle this is out of the scope of this paper. %

\section{Problem Formulation}\label{sec:probForm}

\subsection{Problem Statement}\label{Prob_Form}

Consider a network $\mathcal{G}_d = (\mathcal{V}, \mathcal{E})$. 
Each one of the $n  = | \mathcal{V} |$ nodes is endowed with a scalar quadratic local cost function $f_i : \mathbb{R}^n \mapsto \mathbb{R}$. 
In most cases \cite{2004:Sensors, Tao:2019Survey} a quadratic cost function of the following form is considered: 
\begin{equation}\label{local_cost_functions}
    f_i(z) = \dfrac{1}{2} \alpha_i (z - \rho_i)^2 , 
\end{equation}
where $\alpha_i > 0$, $\rho_i \in \mathbb{R}$ is the demand in node $v_i$ (and in our case is a positive real number) and $z$ is a global optimization parameter that will determine the workload at each node. 

The global cost function is the sum of the local cost functions $f_i : \mathbb{R}^n \mapsto \mathbb{R}$ (shown in \eqref{local_cost_functions}) of every node $v_i \in \mathcal{V}$. 
The main goal of the nodes is to allocate the jobs in order to minimize the global cost function 
\begin{align}\label{opt:1}
z^* =  \arg\min_{z\in \mathcal{Z}} \sum_{v_{i} \in \mathcal{V}} f_i(z) , 
\end{align}
where $\mathcal{Z}$ is the set of feasible values of parameter $z$. 
Optimization problem \eqref{opt:1} can be solved in closed form and $z^*$ is given by
\begin{align}\label{eq:closedform}
z^* =  \frac{\sum_{v_{i} \in \mathcal{V}} \alpha_i \rho_{i}}{\sum_{v_{i} \in \mathcal{V}} \alpha_i}.
\end{align}
Note that if $\alpha_i =1$ for all $v_{i}\in\mathcal{V}$, the solution is the average.

\subsection{Modification of the Optimization Problem}\label{Opt_Prob_Form}

Nodes require to calculate the optimal solution at every optimization step $m$ via a distributed coordination algorithm which relies on the exchange of quantized values and converges after a finite number of time steps. 
The proposed algorithm allows all nodes to balance their CPU utilization (i.e., the same percentage of capacity) during the execution of the tasks, i.e., 
\begin{align}\label{cond:balance}
\frac{w_i^*[m] +u_{i}[m]}{\pi_i^{\max}} &= \frac{w_j^*[m] +u_{j}[m]}{\pi_j^{\max}} \\
&= \frac{\rho[m] + u_{\mathrm{tot}}[m]}{\pi^{\max}}, \ \forall v_{i}, v_{j} \in \mathcal{V}, \nonumber
\end{align}
where $w_i^*[m]$ is the \emph{optimal} workload to be added to server node $v_{i}$ at optimization step $m$, $\pi^{\max} \coloneqq \sum_{v_{i}\in \mathcal{V}} \pi_i^{\max}$ and $u_{\mathrm{tot}}[m]=\sum_{v_{i}\in \mathcal{V}} u_{i}[m]$. For simplicity of exposition, and since we consider a single optimization step, we drop index $m$.
To achieve the requirement set in \eqref{cond:balance}, we need the solution (according to \eqref{eq:closedform}) to be \cite{2020:Themis_Kalyvianaki} 
\begin{align}\label{eq:closedform1}
z^* =  \frac{\sum_{v_{i} \in \mathcal{V}} \pi_i^{\max} \frac{\rho_{i}+u_{i}}{\pi_i^{\max}}}{\sum_{v_{i} \in \mathcal{V}} \pi_i^{\max}} = \frac{\rho + u_{\mathrm{tot}}}{\pi^{\max}}.
\end{align}
Hence, we modify \eqref{local_cost_functions} accordingly. Then, the cost function $f_i(z)$ in \eqref{local_cost_functions} is given by
\begin{align}\label{eq:fiz}
f_i(z) = \frac{1}{2}\pi_i^{\max} \left(z- \frac{\rho_{i}+u_{i}}{\pi_i^{\max}} \right)^2.
\end{align}
In other words, each node computes its proportion of workload and from that it is able to find the workload $w_i^*$ to receive, i.e.,
 \begin{align}\label{eq:optimal_workload}
w_i^*  = \frac{\rho + u_{\mathrm{tot}}}{\pi^{\max}} \pi_i^{\max} - u_{i}.
\end{align}

The solution should be found in a distributed way. 
Specifically, we aim at developing a distributed coordination algorithm to find the solution via the exchange of information only between neighboring nodes. 
The algorithm should rely on processing and transmitting of quantized information while its operation should exhibit finite time convergence.

\section{Preliminaries on Distributed Coordination}\label{sec:prelim_distrCoord}

\subsection{Quantized Average Consensus}\label{Prel_Aver}

The objective of quantized average consensus problems is the development of distributed algorithms which allow nodes to process and transmit quantized information. 
During their operation, each node utilizes short communication packages and eventually obtains a state $q^s$ which is equal to the largest quantized value (but not greater) or the smallest quantized value (but not lower) of the real average $q$ of the initial quantized states, after a finite number of time steps. 

In this paper we consider the case where quantized values are represented by integer\footnote{Following \cite{2007:Basar} we assume that the state of each node is integer valued.
This abstraction subsumes a class of quantization effects (e.g., uniform quantization).} numbers. 
This means that each node in the network is able to obtain a state $q^s$ which is equal to the ceiling $\lceil q \rceil$ or the floor $\lfloor q \rfloor$ of the real average $q$ of the initial quantized states of the nodes, after a finite number of time steps.

Since each node processes and transmits quantized information, we adopt the algorithm in \cite{2021:Rikos_Hadj_Johan}. 
Specifically, the algorithm in \cite{2021:Rikos_Hadj_Johan} is preliminary for our results in this paper and during its operation, each node is able to achieve quantized average consensus after a finite number of time steps. 
We make the following assumption which is necessary for the operation of the algorithm in \cite{2021:Rikos_Hadj_Johan} as well as the operation of our proposed algorithm in this paper. 
More specifically, assumption~\ref{str_digr} below is a necessary condition for each node $v_j$ to be able to calculate the quantized average of the initial values after a finite number of time steps. 

\begin{assum}\label{str_digr}
The communication topology is modeled as a strongly connected digraph. 
\end{assum}

The operation of the algorithm presented in \cite{2021:Rikos_Hadj_Johan}, assumes that each node $v_j$ in the network has an integer initial state $y_j[0] \in \mathbb{Z}$. 
At each time step $k$, each node $v_j \in \mathcal{V}$ maintains its mass variables $y_j[k] \in \mathbb{Z}$ and $z_j[k] \in \mathbb{Z}_{\geq 0}$, and its state variables $y^s_j[k] \in \mathbb{Z}$, $z^s_j[k] \in \mathbb{N}$ and  $q_j^s[k] = \lceil \frac{y_j^s[k]}{z_j^s[k]} \rceil$. 
It updates the values of the mass variables as 
\begin{subequations}\label{Prel_Aver_YZ}
\begin{align}
y_j[k+1] = y_j[k] + \sum_{v_i \in \mathcal{N}_j^-} \mathds{1}_{ji}[k] y_i[k] , \label{subeq:1a} \\
z_j[k+1] = z_j[k] + \sum_{v_i \in \mathcal{N}_j^-} \mathds{1}_{ji}[k] z_i[k] , \label{subeq:1b}
\end{align}
\end{subequations}
where  
\begin{align*}
\mathds{1}_{ji}[k] = 
\begin{cases}
1, & \text{if a message is received at $v_j$ from $v_i$ at $k$,} \\[0.1cm]
0, & \text{otherwise.} 
\end{cases}
\end{align*}
If the following event-triggered condition holds: 
\begin{list3}\label{tr_cond}
\item[(C1):] $z_j[k] > 1$ ,
\end{list3}
then, node $v_j$ updates its state variables as follows: 
\begin{subequations}\label{State_Prel_Aver_YZ}
\begin{align}
z^s_j[k+1] &= z_j[k+1], \\
y^s_j[k+1] &= y_j[k+1], \\
q^s_j[k+1] &= \Bigl \lceil \frac{y^s_j[k]}{z^s_j[k]} \Bigr \rceil .
\end{align}
\end{subequations}
Then, it splits $y_j[k]$ into $z_j[k]$ equal integer pieces (with the exception of some pieces whose value might be greater than others by one). 
It chooses one piece with minimum $y$-value and transmits it to itself, and it transmits each of the remaining $z_j[k]-1$ pieces to randomly selected out-neighbors or to itself. 
Finally, it receives the values $y_i[k]$ and $z_i[k]$ from its in-neighbors, sums them with its stored $y_j[k]$ and $z_j[k]$ values (as described in \eqref{subeq:1a}, \eqref{subeq:1b}) and repeats the operation.

\begin{defn}\label{Definition_Quant_Av}
The system is able to achieve quantized average consensus if, for every $v_j \in \mathcal{V}$, there exists $k_0 \in \mathbb{Z}_+$ so that for every $v_j \in \mathcal{V}$ we have
\begin{equation}\label{alpha_q_no_oscill}
( q^s_j[k] = \lfloor q \rfloor \ \ \text{for} \ \ k \geq k_0 ) \ \ \text{or} \ \ ( q^s_j[k] = \lceil q \rceil \ \ \text{for} \ \ k \geq k_0) ,
\end{equation}
where $q$ is the real average of the initial states defined as: 
\begin{equation}\label{real_av}
q = \frac{\sum_{l=1}^{n}{y_l[0]}}{n} .
\end{equation}
\end{defn}
The following result from~\cite{2021:Rikos_Hadj_Johan} provides an upper bound regarding the number of time steps required for quantized average consensus to be achieved. 

\begin{theorem}[\hspace{-0.00001cm}~\cite{2021:Rikos_Hadj_Johan}]
\label{Conver_Quant_Av}
The iterations in (\ref{Prel_Aver_YZ}) and (\ref{State_Prel_Aver_YZ}) allow the set of nodes to reach quantized average consensus (i.e., state variable $q_j^s$ of each node $v_j \in \mathcal{V}$ fulfils \eqref{alpha_q_no_oscill}) after a finite number of steps. 
Specifically, for any $\varepsilon$, where $0 < \varepsilon < 1$, there exists $k_0 \in \mathbb{Z}_+$, so that with probability $(1-\varepsilon)^{(y^{init}+n)}$ we have
$$ 
( q^s_j[k] = \lfloor q \rfloor \ \ \text{for} \ \ k \geq k_0) \ \ \text{or} \ \ ( q^s_j[k] = \lceil q \rceil \ \ \text{for} \ \ k \geq k_0)  ,
$$
for every $v_j \in \mathcal{V}$, where $q$ fulfills \eqref{real_av} and
\begin{align}
y^{init} & = & \sum_{\{v_j \in \mathcal{V}: y_j[0] > \lceil q \rceil\}} {(y_j[0] - \lceil q \rceil) } \ + \nonumber \\ 
 & & \sum_{\{v_j \in \mathcal{V}: y_j[0] < \lfloor q \rfloor\}} {(\lfloor q \rfloor - y_j[0])} . %
\end{align}
\end{theorem}

\subsection{Synchronous max/min - Consensus}

The $\max$-consensus algorithm computes the maximum value of the network in a finite number of time steps in a distributed fashion \cite{2008:Cortes}. 
For every node $v_{j} \in \mathcal{V}$, if the updates of the node's state are synchronous, then the update rule is:
\begin{align}
x_j[k+1] = \max_{v_{i}\in \mathcal{N}_j^{-} \cup \{v_{j}\}}\{ x_i[k] \}.
\end{align}
It has been shown (see, e.g., \cite[Theorem 5.4]{2013:Giannini}) that the $\max$-consensus algorithm converges to the maximum value among all nodes in a finite number of steps $s$, where $s \leq D$.  
Similar results hold for the $\min$-consensus algorithm.

\section{Quantized CPU Scheduling Algorithm}
\label{sec:distr_algo}

In this section we propose a distributed quantized information exchange algorithm which solves the problem described in Section~\ref{sec:probForm}. 
The proposed algorithm is detailed as Algorithm~\ref{algorithm1} below. 
The distributed algorithm allows each node $v_j$ to calculate the optimal required workload $w^*_j$ shown in \eqref{eq:optimal_workload}, after a finite number of time steps. 
For solving the problem in a distributed way we make the following two assumptions.

\begin{assum}\label{Diam_known}
The diameter of the network $D$ (or an upper bound $D'$) is known to all server nodes $v_j \in \mathcal{V}$.
\end{assum}

\begin{assum}\label{upper_bound_known}
Each server node $v_j \in \mathcal{V}$ has knowledge of an upper bound $\pi^{\mathrm{upper}}$ regarding the total capacity of the network $\pi^{\max}$ (i.e., $\pi^{\mathrm{upper}} \geq \pi^{\max}$, where $\pi^{\max} \coloneqq \sum_{v_{j} \in \mathcal{V}} \pi_j^{\max}$).
\end{assum}

Assumption~\ref{Diam_known} is necessary for coordinating the $\min$- and $\max$-consensus algorithm, such that each node $v_j$ is able to determine whether convergence has been achieved and thus the operation of our proposed algorithm needs to be terminated.

Assumption~\ref{upper_bound_known} is made such that our proposed algorithm allows each node $v_j$ to calculate the correct optimal required workload $w^*_j$ in a finite number of time steps via exchanging quantized information with its neighbors. 
Specifically, each node $v_j$ needs to know $\pi^{\mathrm{upper}}$ (where $\pi^{\mathrm{upper}} \geq \pi^{\max}$) in order to multiply its initial value $y_j[0]$ with $\pi^{\mathrm{upper}}$ so that $y_j[0] > z_j[0]$ (here $z_j[0]$ is a variable used by node $v_j$ to process the value of $y_j[0]$ as it will be seen later in the proposed algorithm). 
Guaranteeing that $y_j[0] > z_j[0]$ is necessary during the operation of our algorithm, so that each node $v_j$ is able to split $y_j[k]$ into $z_j[k]$ equal integer pieces (or with maximum difference between them equal  to $1$) at every time step $k \in \mathbb{N}$.

\begin{varalgorithm}{1}
\caption{Quantized CPU Scheduling Algorithm}
\noindent \textbf{Input:} A strongly connected digraph $\mathcal{G}_d = (\mathcal{V}, \mathcal{E})$ with $n=|\mathcal{V}|$ nodes and $m=|\mathcal{E}|$ edges. 
Each node $v_j\in \mathcal{V}$ has knowledge of $l_j, u_j, D, \pi^{\mathrm{upper}}, \pi_j^{\max} \in \mathbb{Z}$. \\
\textbf{Initialization:} Each node $v_j \in \mathcal{V}$ does the following: 
\begin{list4}
\item[$1)$] Assigns a nonzero probability $b_{lj}$ to each of its outgoing edges $m_{lj}$, where $v_l \in \mathcal{N}^+_j \cup \{v_j\}$, as follows
\begin{align*}
b_{lj} = \left\{ \begin{array}{ll}
         \frac{1}{1 + \mathcal{D}_j^+}, & \mbox{if $l = j$ or $v_{l} \in \mathcal{N}_j^+$,}\\
         0, & \mbox{if $l \neq j$ and $v_{l} \notin \mathcal{N}_j^+$.}\end{array} \right. 
\end{align*} 
\item[$2)$] Sets $y_j[0] := \pi^{\mathrm{upper}} (l_j + u_j)$, $z_j[0] = \pi_j^{\max}$, and $\text{flag}_j = 0$. 
\end{list4} 
\textbf{Iteration:} For $k=1,2,\dots$, each node $v_j \in \mathcal{V}$, does the following: 
\begin{list4} 
\item \textbf{while} $\text{flag}_j = 0$ \textbf{then} 
\begin{list4a}
\item[$1)$] \textbf{if} $k \mod D = 1$ \textbf{then} sets $M_j = \lceil y_j[k] / z_j[k] \rceil$, $m_j = \lfloor y_j[k] / z_j[k] \rfloor$; 
\item[$2)$] broadcasts $M_j$, $m_j$ to every $v_{l} \in \mathcal{N}_j^+$; 
\item[$3)$] receives $M_i$, $m_i$ from every $v_{i} \in \mathcal{N}_j^-$; 
\item[$4)$] sets $M_j = \max_{v_{i} \in \mathcal{N}_j^-\cup \{ v_j \}} \ M_i$, $m_j = \min_{v_{i} \in \mathcal{N}_j^-\cup \{ v_j \}} \ m_i$; 
\item[$5)$] \textbf{if} $z_j[k] > 1$, \textbf{then}
\begin{list4a}
\item[$5.1)$] sets $z^s_j[k] = z_j[k]$, $y^s_j[k] = y_j[k]$, 
$
q^s_j[k] = \Bigl \lceil \frac{y^s_j[k]}{z^s_j[k]} \Bigr \rceil \ ;
$
\item[$5.2)$] sets (i) $mas^y[k] = y_j[k]$, $mas^z[k] = z_j[k]$; (ii) $c^{y}_{lj}[k] = 0$, $c^{z}_{lj}[k] = 0$, for every $v_l \in \mathcal{N}^+_j \cup \{v_j\}$; (iii) $\delta = \lfloor mas^y[k] / mas^z[k] \rfloor$, $mas^{rem}[k]= y_j[k] - \delta \ mas^z[k]$;  
\item[$5.3)$] \textbf{while} $mas^z[k] > 1$, \textbf{then} 
\begin{list4a}
\item[$5.3a)$] chooses $v_l \in \mathcal{N}^+_j \cup \{v_j\}$ randomly according to $b_{lj}$; 
\item[$5.3b)$] sets (i) $c^{z}_{lj}[k] := c^{z}_{lj}[k] + 1$, $c^{y}_{lj}[k] := c^{y}_{lj}[k] + \delta$; (ii) $mas^z[k] := mas^z[k] - 1$, $mas^y[k] := mas^y[k] - \delta$. 
\item[$5.3c)$] If $mas^{rem}[k] > 1$, sets $c^{y}_{lj}[k] := c^{y}_{lj}[k] + 1$, $mas^{rem}[k] := mas^{rem}[k]- 1$; 
\end{list4a}
\item[$5.4)$] sets $c^{y}_{jj}[k] := c^{y}_{jj}[k] + mas^y[k]$, $c^{z}_{jj}[k] := c^{z}_{jj}[k] + mas^z[k]$; 
\item[$5.5)$] for every $v_l \in \mathcal{N}^+_j$, if $c^{z}_{lj}[k] > 0$ transmits $c^{y}_{lj}[k]$, $c^{z}_{lj}[k]$ to out-neighbor $v_l$; 
\end{list4a}
\item \textbf{else if} $z_j[k] \leq 1$, sets $c^{y}_{jj}[k] = y[k]$, $c^{z}_{jj}[k] = z[k]$; 
\item[$6)$] receives $c^{y}_{ji}[k]$, $c^{z}_{ji}[k]$ from $v_i \in \mathcal{N}_j^-$ and sets 
\begin{equation}\label{no_del_eq_y_no_oscil}
y_j[k+1] = c^{y}_{jj}[k] + \sum_{v_i \in \mathcal{N}_j^-}  w_{ji}[k] \ c^{y}_{ji}[k] ,
\end{equation} 
\begin{equation}\label{no_del_eq_z_no_oscil}
z_j[k+1] = c^{z}_{jj}[k] + \sum_{v_i \in \mathcal{N}_j^-} w_{ji}[k] \ c^{z}_{ji}[k] ,
\end{equation}
where $w_{ji}[k] = 1$ if node $v_j$ receives $c^{y}_{ji}[k]$, $c^{z}_{ji}[k]$ from $v_i \in \mathcal{N}_j^-$ at iteration $k$ (otherwise $w_{ji}[k] = 0$); 
\item[$7)$] \textbf{if} $k \mod D = 0$ \textbf{then}, \textbf{if} $M_j - m_j \leq 1$ \textbf{then} sets $w_j^* = \lceil q^s_j[k] (\pi_j^{\max} / \pi^{\mathrm{upper}}) \rceil$ and $\text{flag}_j = 1$.
\end{list4a}
\end{list4}
\textbf{Output:} \eqref{cond:balance} holds for every $v_j \in \mathcal{V}$. 
\label{algorithm1}
\end{varalgorithm}

\begin{remark}
It is interesting to note here that Algorithm~\ref{algorithm1} is based on similar principles as the algorithm presented in \cite{2015:Cady}, which executes the ratio-consensus algorithm \cite{2010:christoforos} along with $\min-$ and $\max-$consensus iterations \cite{2008:Cortes}. 
Specifically, during the operation in \cite{2010:christoforos}, each node maintains two real valued variables and updates them by executing two parallel iterations. Then, each node is able to calculate the real average of the initial states asymptotically as the ratio of these two variables. 
Furthermore, by performing $\min$- and $\max$-consensus \cite{2008:Cortes} every $D$ time steps, each node is able to determine during which time step $k_0$ its state is within $\varepsilon$ to the state of every other node (i.e., their difference is less or equal to $\varepsilon$). 
Overall, \cite{2015:Cady} allowed the nodes in the network to calculate the real average of their initial states and then terminate their operation according to a distributed stopping criterion. 
Nevertheless, compared to \cite{2015:Cady}, Algorithm~\ref{algorithm1} has significant differences due to its quantized nature. 
These differences mainly focus on (i) the underlying process for calculating the quantized average of the initial states via the exchange of quantized messages, and (ii) the distributed stopping mechanism designed explicitly for quantized information exchange algorithms.  
Specifically, during the operation of Algorithm~\ref{algorithm1}, the underlying process for calculating the quantized average of the initial states is based on \cite{2021:Rikos_Hadj_Johan}. 
This means that each node maintains two integer valued variables and updates them by executing two parallel iterations, where it splits them into integer equal pieces (or with maximum difference equal to $1$) and transmits them to randomly chosen out-neighbors. 
Then, each node calculates the quantized average of the initial states in a finite number of time steps as the ceiling of the ratio of these two variables. 
Furthermore, the distributed stopping mechanism is based on performing $\min$- and $\max$-consensus every $D$ time steps, where the $\min$- and $\max$-values are initialized as the floor and the ceiling of the ratio of the two integer valued variables it maintains. 
The $\min$- and $\max$-consensus operation converges once the $\min$-values are within $1$ of the $\max$-values (i.e., their difference is less or equal to $1$) which means that the state of every node is within $1$ to the state of every other node. 
As a result, Algorithm~\ref{algorithm1} allows the nodes to calculate the quantized average of the initial states and, by utilizing the distributed stopping mechanism, to determine whether convergence has been achieved, and, thus whether the operation can be terminated. 
\end{remark}

Next, we show that, during the operation of Algorithm~\ref{algorithm1}, each node $v_j$ is able to (i)  calculate the optimal required workload $w^*_j$ (shown in \eqref{eq:optimal_workload}) after a finite number of time steps, and (ii) after calculating $w^*_j$ terminate its operation.

\begin{theorem}
\label{thm:main}
Consider a strongly connected digraph $\mathcal{G}_d = (\mathcal{V}, \mathcal{E})$ with $n=|\mathcal{V}|$ nodes and $m=|\mathcal{E}|$ edges and $y_j[0] = \pi^{\mathrm{upper}} (l_j + u_j)$, $z_j[0] = \pi_j^{\max}$ where $l_j, u_j, \pi^{\mathrm{upper}}, \pi_j^{\max} \in \mathbb{N}$ for every node $v_j \in \mathcal{V}$ at time step $k=0$. 
Suppose that each node $v_j \in \mathcal{V}$ follows the Initialization and Iteration steps as described in Algorithm~\ref{algorithm1}.
For any $\varepsilon$, where $0 < \varepsilon < 1$, there exists $k_0 \in \mathbb{N}$, so that for each node $v_j$ it holds 
$$
w_j^* = \lceil q^{\mathrm{tasks}} (\pi_j^{\max} / \pi^{\mathrm{upper}}) \rceil = \frac{\rho + u_{\mathrm{tot}}}{\pi^{\max}} \pi_i^{\max} - u_{i} , 
$$
with probability $(1-\varepsilon)^{(y^{init}+n)}$ where 
\begin{equation}\label{initial_dist_processors}
q^{\mathrm{tasks}} = \pi^{\mathrm{upper}} \frac{\sum_{v_{j} \in \mathcal{V}} (l_j + u_j)}{\sum_{v_{j} \in \mathcal{V}} \pi_j^{\max}} , 
\end{equation} 
and 
\begin{align}
y^{init} & = & \sum_{\{v_j \in \mathcal{V}: y_j[0] > \lceil q^{\mathrm{tasks}} \rceil\}} {(y_j[0] - \lceil q^{\mathrm{tasks}} \rceil) } \ + \nonumber \\ 
 & & \sum_{\{v_j \in \mathcal{V}: y_j[0] < \lfloor q^{\mathrm{tasks}} \rfloor\}} {(\lfloor q^{\mathrm{tasks}} \rfloor - y_j[0])} , \label{initial_distance_no_oscill}
\end{align}
is the total initial state error (i.e., $y^{init}$ is the sum of the differences between (i) the value $\lceil q^{\mathrm{tasks}} \rceil$ and the initial state $y_j[0]$ of each node $v_j$ that has an initial state higher than the ceiling of $q^{\mathrm{tasks}}$ and (ii) the value $\lfloor q^{\mathrm{tasks}} \rfloor$ and the initial state $y_j[0]$ of each node $v_j$ that has an initial state less than the floor of $q^{\mathrm{tasks}}$). 

This means that each node $v_j$ is able to (i) calculate the optimal required workload $w^*_j$ (shown in \eqref{eq:optimal_workload}) after a finite number of time steps $k_0$ with probability $(1-\varepsilon)^{(y^{init}+n)}$ and (ii) after calculating $w^*_j$, terminate its operation. 
\end{theorem}

\begin{proof}
See Appendix~\ref{appendix:A}.
\end{proof}

\section{Simulation Results} \label{sec:results}

In this section, we present simulation results to illustrate the behavior of our proposed distributed algorithm.
In the first section, we present a random graph of $200$ nodes and show how the states of the node converge.
In the second section, we present a more quantitative analysis over a larger set of network sizes which would be more applicable to practical deployments, such as in modern data-centers.
To the best of our knowledge, this is the first work that tries to tackle the problem of converging using quantized values at that scale while also providing a thorough evaluation accompanied with strong theoretical guarantees.
All experiments were performed on a workstation using an AMD 3970X CPU with 32 cores, 256GB of 3600 MHz DDR4 RAM, and MatLab R2021a (build 9.10.0.1602886).
To foster reproducibility, all of the code, datasets, and experiments will be made publicly available.\footnote{\url{https://github.com/andylamp/federated-quantized-ratio-consensus}}.

\subsection{Evaluation over a Small Scale Network}
\label{Small_Network_Simulation}
In this section we show how the states of the nodes converge during the iteration.
The network is large enough to provide valuable insights, yet able to be visualised concretely.
The network in this example comprised $200$ nodes and was randomly generated (an edge between a pair of nodes exists with probability $0.5$).
This process resulted in a digraph that had a diameter equal to $2$. 
Small digraph diameters are indicative on data-center topologies and are normally preferred due to their locality and the benefit of having few hops between each node~\cite{besta2014slim}.
The upper bound $\pi^{\mathrm{upper}}$ of the total capacity is $1000$ and the workload $l_j$ of each node $v_j$ was generated using a random distribution uniformly picked within the range $[1, 100]$. 
The node capacities $\pi_j^{\max}$ in this experiment were set to either $100$ or $300$ for even and odd node numbers respectively.
Our simulation results are shown in Fig.~\ref{eval-small-example}, which depicts the load per n
ode according to its processing capacity.
We can see that the network converges monotonically within a few iterations without being affected by value oscillations or ambiguities.

\begin{figure}[t]
    \centering
    \includegraphics[width=\linewidth]{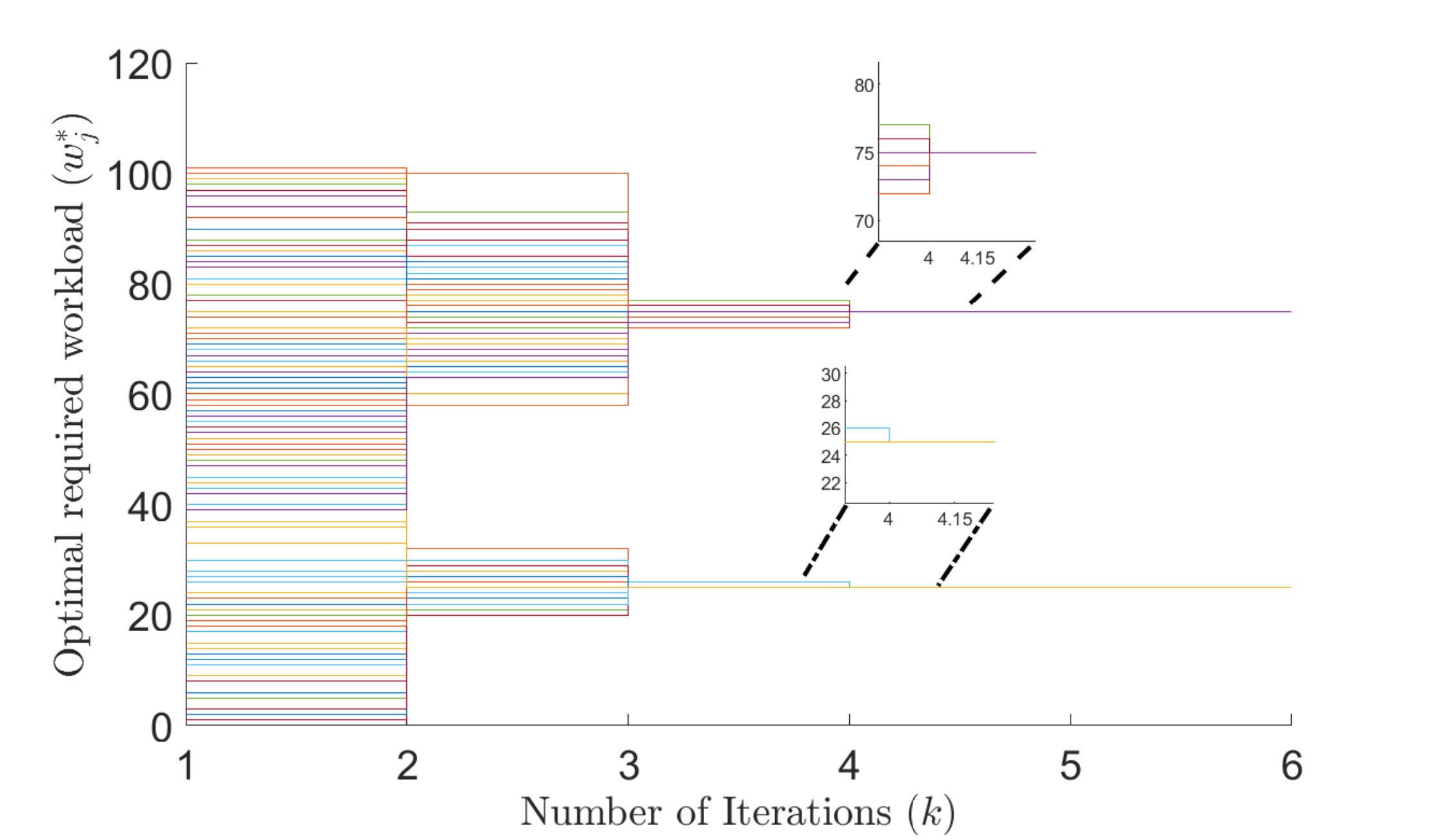}
    \caption{Execution of Algorithm~\ref{algorithm1} over a random network comprised of $200$ nodes having a diameter equal to $2$. We see that the network converges in less than $10$ iterations while having no oscillations. }
    \label{eval-small-example}
\end{figure}

\subsection{Data Center Scale Evaluation}
\label{Data_Center_Simulation}

Our previous analysis dealt with a quantitative example showing the weights for all nodes involved across all iterations.
Here, we present a large scale evaluation of networks over a wide gamut of sizes.
Concretely, we evaluate our proposed scheme on networks sized from $20$ nodes up to $10000$ nodes.
The topologies are randomly generated and result in digraphs that have a diameter from $2$ to $10$.
As we previously mentioned, such digraph diameters are indicative of practical data-center deployments.
We evaluated each network size across $50$ trials and the aggregated values were averaged out before plotting.
The upper bound of the total capacity $\pi^{\mathrm{upper}}$ for all trials was set to $1000$ and the workloads were generated similarly to the previous example.

We start by presenting the iterations required for all of these networks to converge; these results are shown in Fig.~\ref{iterations-to-converge}.
We can see that across \textit{all} network sizes our scheme required less than $40$ iterations to converge.
Another interesting observation is that as network sizes grow, the number of iterations to converge \textit{drops}.
However, as we will see later this does not necessarily mean an improvement in the overall runtime of the algorithm.

\begin{figure}[t]
    \centering
    \includegraphics[width=\linewidth]{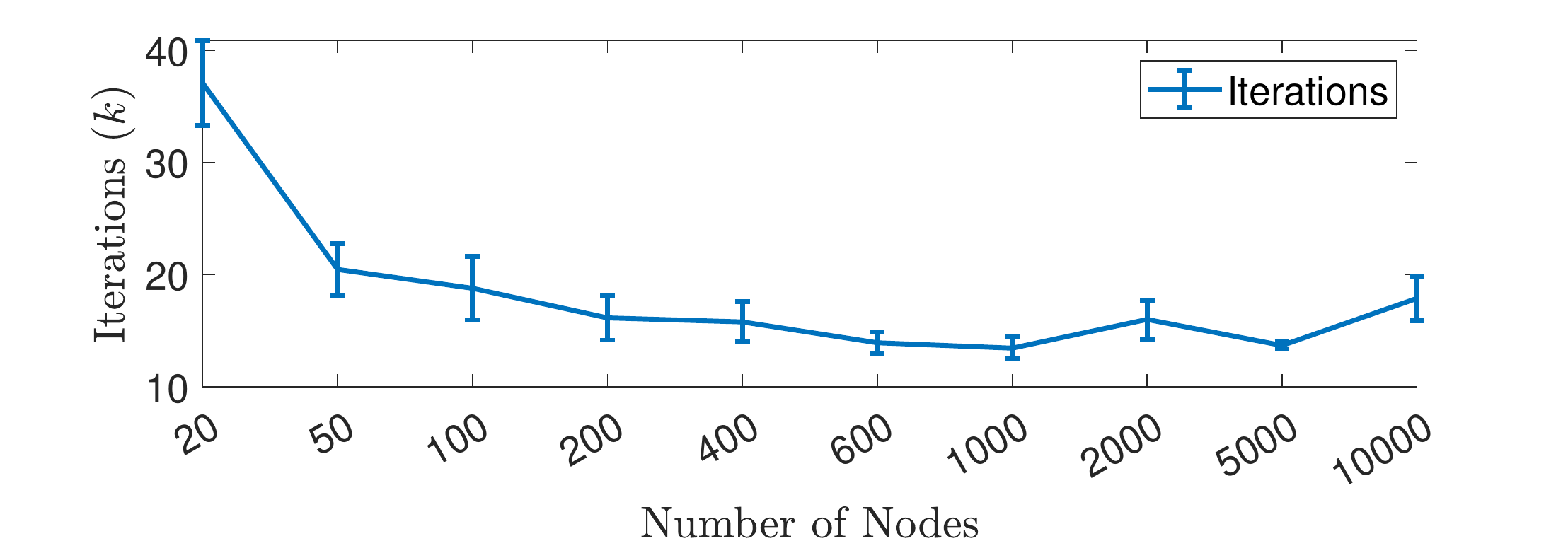}
    \caption{Required iterations for convergence of different network sizes during the operation of Algorithm~\ref{algorithm1} along with their error bars. Each network size is evaluated across $50$ trials and the aggregated values were averaged out before plotting.}
    \label{iterations-to-converge}
\end{figure}

Next, we present the actual time it took for our algorithm to converge.
In the previous digraph we showed the number of iterations required to converge, however that does not tell the whole story. 
In larger networks each iteration takes \textit{longer} as the amount of nodes to communicate increases; hence, prolonging the duration of each round.
This can be very evidently seen in the absolute time it took for the trials to complete for smaller and larger networks, the results of which are presented in Fig.~\ref{time-to-converge}.

\begin{figure}[t]
    \centering
    \includegraphics[width=\linewidth]{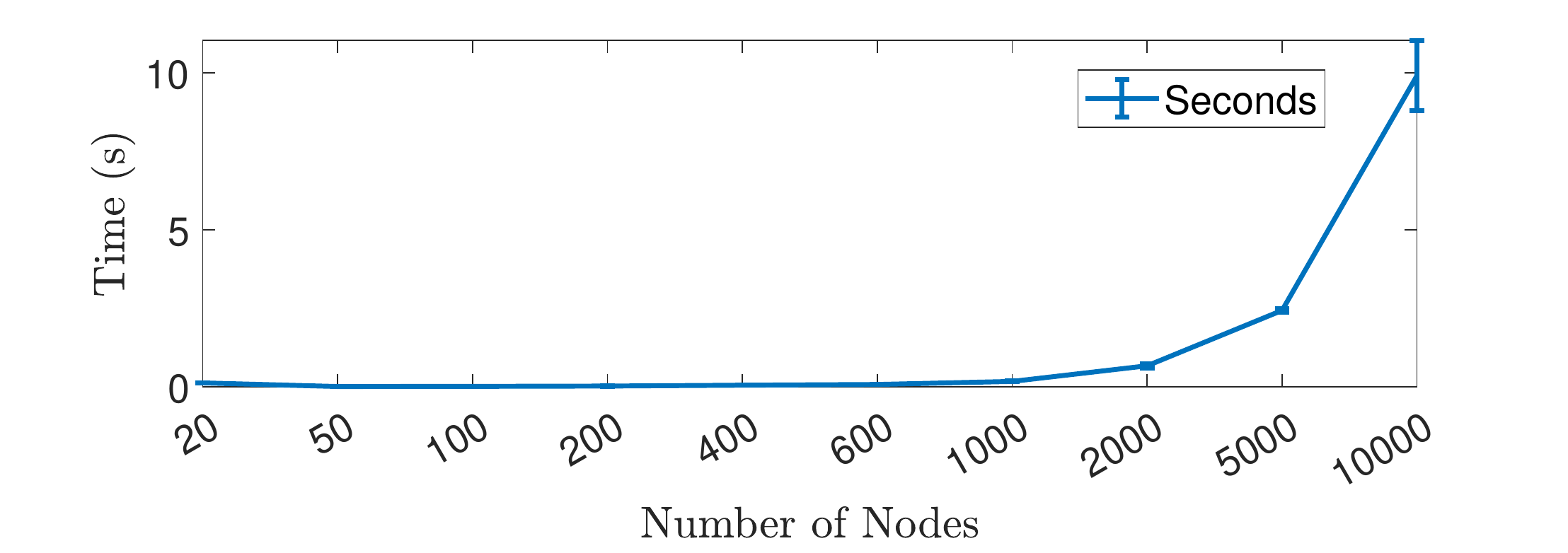}
    \caption{Required time, in seconds, for Algorithm~\ref{algorithm1} to converge along with error bars. 
    The plot is averaged over $50$ trials for each network size. We see that networks below $1000$ nodes converge in less than a second.}
    \label{time-to-converge}
\end{figure}

We note, however, that practically our scheme for networks up to a thousand nodes converges fast (in \textit{less than a second}).
We now move to a more qualitative evaluation of the performance presenting the aggregated converge statistics in Fig.~\ref{converge-stats}.
We present the iteration in which the first node converged within the network, even if some are still divergent.
Additionally, we track the max which is when the last node of the network converged.
Interestingly, their window, which is the absolute difference between max and min, and the average converge iteration shows us how the network behaves as it grows in size.
We can see that as the network size grows the mean tends to level out with the min curve, leaving a few divergent nodes that need more iterations to converge.
Another revealing observation, is that for network sizes larger than $50$ nodes the mean iterations for convergence is \textit{less than $10$}. 
These findings could help understand how these networks behave across different parameters, topologies, and workloads also offering valuable insights for practical deployment.

\begin{figure}[t]
    \centering
    \includegraphics[width=\linewidth]{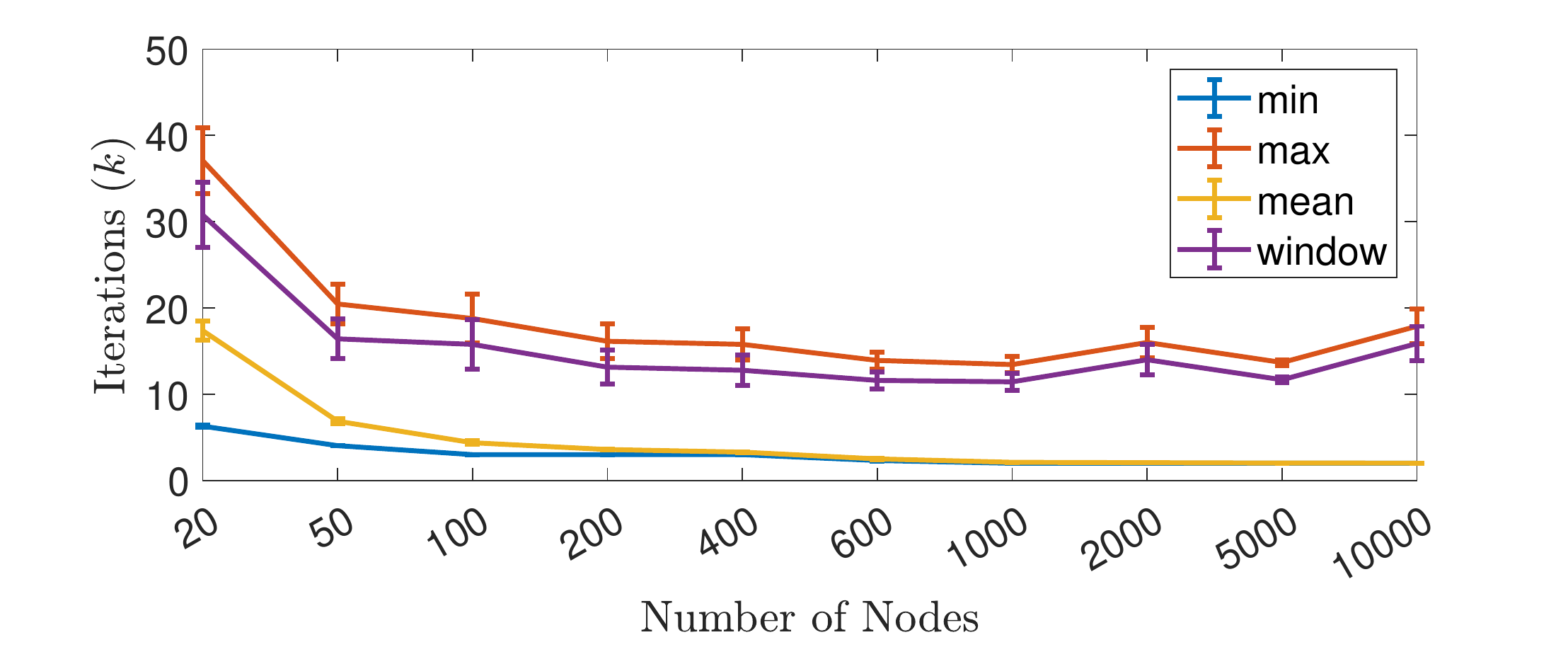}
    \caption{Convergence statistics across different network sizes from $20$ nodes up to $10000$ nodes along with their error bars. The plots are evaluated over $50$ trials and the aggregate values were averaged out before plotting.
    The min indicates when the \textit{first} node in a network converges to its final value.
    The max shows when the \textit{last} node in the network converged and by extension signals the termination of our algorithm.
    The mean shows the average number of iterations required for all of the nodes converge while the window presents the average difference between the minimum and maximum converge values. 
    }
    \label{converge-stats}
\end{figure}

\section{Conclusions and Future Directions}\label{sec:conclusions}

In this paper, we have considered the problem of optimal task scheduling for data centers. 
We proposed a fast distributed iterative algorithm which operates over a large scale network and allows each of the interconnected nodes to reach agreement in a finite number of time steps. 
In the context of task scheduling, we showed that our algorithm converges to the exact optimal scheduling plan in a finite number of time steps and then it exhibits its distributed stopping capability. 
Furthermore, the operation of our algorithm is event-based and relies on the exchange of quantized values between nodes in the network. 
Finally, we have demonstrated the performance of our proposed algorithm and by using extensive empirical evaluations, we have shown the algorithm's fast convergence.

Although the quantized algorithm works in an asynchronous manner as well and converges in a finite number of steps, a termination mechanism should be deployed that would allow a fully asynchronous coordination for solving the optimization problem, which is desirable in large-scale networks.

\appendices
\section{Proof of Theorem~\ref{thm:main}}
\label{appendix:A}

We first consider Lemma~\ref{Lemma1_prob}, which is necessary for our subsequent development. 

\begin{lemma}[\hspace{-0.00001cm}\cite{2020:RikosHadj_IFAC}]
\label{Lemma1_prob}
Consider a strongly connected digraph $\mathcal{G}_d = (\mathcal{V}, \mathcal{E})$ with $n=|\mathcal{V}|$ nodes and $m=|\mathcal{E}|$ edges.
Suppose that each node $v_j$ assigns a nonzero probability $b_{lj}$ to each of its outgoing edges $m_{lj}$, where $v_l \in \mathcal{N}^+_j \cup \{v_j\}$, as follows  
\begin{align*}
b_{lj} = \left\{ \begin{array}{ll}
         \frac{1}{1 + \mathcal{D}_j^+}, & \mbox{if $l = j$ or $v_{l} \in \mathcal{N}_j^+$,}\\
         0, & \mbox{if $l \neq j$ and $v_{l} \notin \mathcal{N}_j^+$.}\end{array} \right. 
\end{align*}
At time step $k=0$, node $v_j$ holds a ``token" while the other nodes $v_l \in \mathcal{V} - \{ v_j \}$ do not. 
Each node $v_j$ transmits the ``token'' (if it has it, otherwise it performs no transmission) according to the nonzero probability $b_{lj}$ it assigned to its outgoing edges $m_{lj}$. 
The probability $P^{n-1}_{T_i}$ that the token is at node $v_i$ after $n-1$ time steps satisfies 
$$
P^{n-1}_{T_i} \geq (1+\mathcal{D}^+_{max})^{-(n-1)} > 0 , 
$$
where $\mathcal{D}^+_{max} = \max_{v_j \in \mathcal{V}} \mathcal{D}^+_{j}$. 
\end{lemma}

Now we consider Theorem~\ref{Theorem2_Alg2_Converge} which is necessary for analyzing Algorithm~\ref{algorithm1}. 
Due to space considerations we provide a sketch of the proof. 
The proof is an adaptation of the proof in Theorem~$1$ in \cite{2021:Rikos_Hadj_Johan} and is included here for completeness.

\begin{theorem}\label{Theorem2_Alg2_Converge}
Consider a strongly connected digraph $\mathcal{G}_d = (\mathcal{V}, \mathcal{E})$ with $n=|\mathcal{V}|$ nodes and $m=|\mathcal{E}|$ edges. 
At time step $k=0$, each node $v_j$ sets $y_j[0] := \pi^{\mathrm{upper}} (l_j + u_j)$ and $z_j[0] = \pi_j^{\max}$ where $l_j, u_j, \pi^{\mathrm{upper}}, \pi_j^{\max} \in \mathbb{N}$. 
Suppose that each node $v_j \in \mathcal{V}$ follows the Initialization and Iteration steps as described in Algorithm~\ref{algorithm1}. 
For any $\varepsilon$, where $0 < \varepsilon < 1$, there exists $k_0 \in \mathbb{N}$, so that with probability $(1-\varepsilon)^{(y^{init}+n)}$ we have
\begin{equation}\label{tasks_convergence}
( q^s_j[k] = \lfloor q^{\mathrm{tasks}} \rfloor \ , \ \ k \geq k_0) \ \ \mathrm{or} \ \ ( q^s_j[k] = \lceil q^{\mathrm{tasks}} \rceil \ , \ \ k \geq k_0)  ,
\end{equation}
for every $v_j \in \mathcal{V}$, where $q^{\mathrm{tasks}}$, $y^{init}$ fulfill \eqref{initial_dist_processors},  \eqref{initial_distance_no_oscill}. 
\end{theorem}

\begin{proof}
The operation of Algorithm~\ref{algorithm1} can be interpreted as the ``random walk'' of $\pi^{\max} - n$ ``tokens'' in a Markov chain, where $\pi^{\max} = \sum_{v_{j} \in \mathcal{V}} \pi_j^{\max}$ and $n=|\mathcal{V}|$. 
Specifically, at time step $k=0$, node $v_j$ holds $\pi_j^{\max}$ ``tokens". 
One token is $T_j^{ins}$ and is stationary, whereas the other $\pi_j^{\max}-1$ tokens are $T_j^{out, \vartheta}$, where $\vartheta = 1, 2, ..., \pi_j^{\max} - 1$, and perform independent random walks. 
Each token $T_j^{ins}$ and $T_j^{out, \vartheta}$ contains a pair of values $y_j^{ins}[k]$, $z_j^{ins}[k]$, and $y_j^{out, \vartheta}[k]$, $z_j^{out, \vartheta}[k]$, where $\vartheta = 1, 2, ..., \pi_j^{\max} - 1$, respectively. 
Initially, we have (i) $y_j^{ins}[0] = \lceil y_j[0] / z_j[0] \rceil$, (ii) $y_j^{out, \vartheta}[0] = \lceil y_j[0] / z_j[0] \rceil$ or $y_j^{out, \vartheta}[0] = \lfloor y_j[0] / z_j[0] \rfloor$ and (iii) $z_j^{ins}[0] = z_j^{out, \vartheta}[0] = 1$ for $\vartheta = 1, 2, ..., \pi_j^{\max} - 1$, such that 
$
y_j^{ins}[0] + \sum_{\vartheta = 1}^{\pi_j^{\max} - 1} y_j^{out, \vartheta}[0] = y_j[0], 
$
and
$
z_j^{ins}[0] + \sum_{\vartheta = 1}^{\pi_j^{\max} - 1} z_j^{out, \vartheta}[0] = z_j[0] . 
$
At each time step $k$, each node $v_j$ keeps the token $T_j^{ins}$ (i.e., it never transmits it) while it transmits the tokens $T_j^{out, \vartheta}$, where $\vartheta = 1, 2, ..., \pi_j^{\max} - 1$, independently to out-neighbors according to the nonzero probability $b_{lj}$ it assigned to its outgoing edges $m_{lj}$ during the Initialization Steps. 
If $v_j$ receives one or more tokens $T_i^{out, \vartheta}$ from its in-neighbors $v_i$ the values $y_i^{out, \vartheta}[k]$ and $y_j^{ins}[k]$ become equal (or with maximum difference equal to $1$); then $v_j$ transmits each received token $T_i^{out, \vartheta}$ to a randomly selected out-neighbor according to the nonzero probability $b_{lj}$ it assigned to its outgoing edges $m_{lj}$. 
Note here that during the operation of Algorithm~\ref{algorithm1} we have 
\begin{equation}\label{sum_preserve}
\sum_{j=1}^n \sum_{\vartheta = 1}^{\pi_j^{\max} - 1} y^{out, \vartheta}_j[k] + \sum_{j=1}^n y^{ins}_j[k] = \sum_{j=1}^n y_j[0] , \ \forall k \in \mathbb{Z}_+ .
\end{equation}

The main idea of this proof is that one token $T^{out, \vartheta}_{\lambda}$ visits a specific node $v_i$ (for which it holds $| y^{out, \vartheta}_{\lambda} - y^{ins}_i | > 1$) and obtains equal values $y$ (or with maximum difference between them equal to $1$) with the token $T^{ins}_i$ which is kept in node $v_i$. 
Thus, we analyze the required time steps for the specific token $T^{out, \vartheta}_{\lambda}$ (which performs random walk) to visit node $v_i$ according to a probability. 
Note here that if each token $T^{out, \vartheta}_{\lambda}$ visits each node $v_i$, $y^{init}$ times then every token in the network (the tokens performing random walk and the stationary tokens) obtains $y$ value equal to $\lfloor q^{\mathrm{tasks}} \rfloor$ or $\lceil q^{\mathrm{tasks}} \rceil$.  

From Lemma~\ref{Lemma1_prob} we have that the probability $P^{n-1}_{T^{out}}$ that ``the specific token $T_{\lambda}^{out, \vartheta}$ is at node $v_i$ after $n-1$ time steps'' is 
\begin{equation}\label{lowerProf_no_oscil}
P^{n-1}_{T^{out}} \geq (1+\mathcal{D}^+_{max})^{-(n-1)} .
\end{equation}
This means that the probability $P^{n-1}_{N\_ T^{out}}$ that ``the specific token $T_{\lambda}^{out, \vartheta}$ has not visited node $v_i$ after $n-1$ time steps'' is
\begin{equation}\label{lowerProf_not_no_oscil}
P^{n-1}_{N\_ T^{out}} \leq 1- (1+\mathcal{D}^+_{max})^{-(n-1)} .
\end{equation}
By extending this analysis, we can state that for any $\epsilon$, where $0 < \varepsilon < 1$ and after $\tau(n-1)$ time steps where
\begin{equation}\label{windows_for_conv_1_no_oscil}
\tau \geq \Big \lceil \dfrac{\log{\epsilon}}{\log{(1 - (1+\mathcal{D}^+_{max})^{-(n-1)})}} \Big \rceil ,
\end{equation}
the probability $P^{\tau}_{N\_ T^{out}}$ that ``the specific token $T_{\lambda}^{out, \vartheta}$ has not visited node $v_i$ after $\tau (n-1)$ time steps'' is
\begin{equation}\label{ProbNotMeet_after_t_no_oscil}
P^{\tau}_{N\_ T^{out}} \leq [P^{n-1}_{N\_ T^{out}}]^{\tau} \leq \epsilon .
\end{equation}
This means that after $\tau (n-1)$ time steps, where $\tau$ fulfills \eqref{windows_for_conv_1_no_oscil}, the probability that ``the specific token $T_{\lambda}^{out, \vartheta}$ has visited node $v_i$ after $\tau (n-1)$ time steps'' is equal to $1-\epsilon$.

Thus, by extending this analysis, for $k \geq (y^{init} + n) \tau (n-1)$, where $y^{init}$ fulfills \eqref{initial_distance_no_oscill} and $\tau$ fulfills \eqref{windows_for_conv_1_no_oscil}, we have 
\begin{equation}\nonumber
( q^s_j[k] = \lfloor q^{\mathrm{tasks}} \rfloor ) \ \ \mathrm{or} \ \ ( q^s_j[k] = \lceil q^{\mathrm{tasks}} \rceil )  ,
\end{equation}
with probability $(1-\varepsilon)^{(y^{init} + n)}$, for every $v_j \in \mathcal{V}$. 
\end{proof}

We are now ready to present the proof of Theorem~\ref{thm:main}. 

\begin{proof}
From Theorem~\ref{Theorem2_Alg2_Converge} we have that the operation of Algorithm~\ref{algorithm1} can be interpreted as the ``random walk'' of $\pi^{\max} - n$ ``tokens'' in a Markov chain, where $\pi^{\max} = \sum_{v_{j} \in \mathcal{V}} \pi_j^{\max}$ and $n=|\mathcal{V}|$. 
Furthermore, we also have that $n$ ``tokens'' remain stationary, one token at each node. 
Each of the $\pi^{\max} - n$ tokens contains a pair of values $y^{out, \vartheta}[k]$, $z^{out, \vartheta}[k]$, where $\vartheta = 1, 2, ..., \pi^{\max} - n$ and each of the $n$ stationary tokens contains a pair of values $y^{ins}[k]$, $z^{ins}[k]$. 
From Theorem~\ref{Theorem2_Alg2_Converge} we have that after $(y^{init} + n) \tau (n-1)$ time steps, where $y^{init}$ fulfills \eqref{initial_distance_no_oscill} and $\tau$ fulfills \eqref{windows_for_conv_1_no_oscil}, the state $q^s_j[k]$ of each node $v_j$ becomes  
$
q^s_j[k] = \lfloor q^{\mathrm{tasks}} \rfloor \ \ \mathrm{or} \ \ q^s_j[k] = \lceil q^{\mathrm{tasks}} \rceil \ , 
$
with probability $(1-\varepsilon)^{(y^{init} + n)}$, where $0 < \varepsilon < 1$, and $q^{\mathrm{tasks}}$ fulfills \eqref{initial_dist_processors}. 
This means that, after $(y^{init} + n) \tau (n-1)$ time steps, where $y^{init}$ fulfills \eqref{initial_distance_no_oscill} and $\tau$ fulfills \eqref{windows_for_conv_1_no_oscil}, for each of the $\pi^{\max} - n$ tokens in the network it holds that $y^{out, \vartheta}[k] = \lfloor q^{\mathrm{tasks}} \rfloor$ or $y^{out, \vartheta}[k] = \lceil q^{\mathrm{tasks}} \rceil$, $\vartheta = 1, 2, ..., \pi^{\max} - n$, while for each of the $n$ stationary tokens in the network it also holds that $y^{ins}[k] = \lfloor q^{\mathrm{tasks}} \rfloor$ or $y^{ins}[k] = \lceil q^{\mathrm{tasks}} \rceil$, with probability $(1-\varepsilon)^{(y^{init} + n)}$, where $0 < \varepsilon < 1$. 
Specifically, the $y$ value of every token in the network is equal either to $\lfloor q^{\mathrm{tasks}} \rfloor$ or $\lceil q^{\mathrm{tasks}} \rceil$ after $(y^{init} + n) \tau (n-1)$ time steps with probability $(1-\varepsilon)^{(y^{init} + n)}$, where $0 < \varepsilon < 1$. 

During the operation of Algorithm~\ref{algorithm1}, every $D$ time steps each node $v_j$ re-initializes its voting variables $M_j$, $m_j$ to be $M_j = \lceil y_j[k] / z_j[k] \rceil$, $m_j = \lfloor y_j[k] / z_j[k] \rfloor$. 
After $(y^{init} + n) \tau (n-1)$ time steps the value $q^s_j[k]$ of each node $v_j$ is equal to $\lfloor q^{\mathrm{tasks}} \rfloor$ or $\lceil q^{\mathrm{tasks}} \rceil$, with probability $(1-\varepsilon)^{(y^{init} + n)}$, where $0 < \varepsilon < 1$. 
This means that $M_j$, $m_j$ are re-initialized to be equal to $M_j = \lfloor q^{\mathrm{tasks}} \rfloor$ or $M_j = \lceil q^{\mathrm{tasks}} \rceil$ and $m_j = \lfloor q^{\mathrm{tasks}} \rfloor$ or $m_j = \lceil q^{\mathrm{tasks}} \rceil$ after $\lceil ((y^{init} + n) \tau (n-1) / D) \rceil D$ time steps with probability $(1-\varepsilon)^{(y^{init} + n)}$, where $0 < \varepsilon < 1$. 
After an additional number of $D$ time steps the variables $M_j$, $m_j$ of each node are updated to $M_j = \lfloor q^{\mathrm{tasks}} \rfloor$ and $m_j = \lfloor q^{\mathrm{tasks}} \rfloor$ (since, the $\max-$consensus algorithm \cite{2008:Cortes} converges after $D$ time steps). 
Thus, the condition $M_j - m_j \leq 1$ holds for every node $v_j$. 
This means that every node $v_j$ calculates the optimal required workload $w_j^*$ and terminates its operation. 
As a result, we have that after $\lceil ((y^{init} + n) \tau (n-1) / D) \rceil D + D$ time steps each node $v_j$ calculates the optimal required workload $w_j^* = \lceil q^{\mathrm{tasks}} (\pi_j^{\max} / \pi^{\mathrm{upper}}) \rceil$ with probability $(1-\varepsilon)^{(y^{init} + n)}$, where $0 < \varepsilon < 1$. 
\end{proof}

\bibliographystyle{IEEEtran}
\bibliography{bibliography}

\end{document}